\documentclass{sigplanconf}
\usepackage{amsmath}
\usepackage{amssymb}
\usepackage{graphicx}
\usepackage{subfigure}
\usepackage{listings}
\usepackage{enumitem}
\setlist{nosep}
\usepackage{amsthm}
\usepackage[pdftitle={lospre in linear time},pdfauthor={Philipp Klaus Krause}]{hyperref}

\renewcommand{\floatpagefraction}{0.95}

\newtheorem{definition}{Definition}
\newtheorem{lemma}{Lemma}
\newtheorem{theorem}{Theorem}
\newcommand{\tw}{\operatorname{tw}}
\newcommand{\coloneqq}{\mathrel{\mathop:}=}
\let\emptyset\varnothing

\makeatletter
\def\@copyrightspace{\relax}
\makeatother

\begin{document}

\setlength{\pdfpageheight}{\paperheight}
\setlength{\pdfpagewidth}{\paperwidth}

\title{lospre in linear time}

\authorinfo{Philipp Klaus Krause}
           {Albert-Ludwigs-Universität Freiburg}
           {krauseph@informatik.uni-freiburg.de}


\maketitle

\begin{abstract}
Lifetime-optimal speculative partial redundancy elimination (lospre) is the most advanced currently known redundancy elimination technique. It subsumes many previously known approaches, such as common subexpression elimination, global common subexpression elimination, and loop-invariant code motion. However, previously known lospre algorithms have high time complexity; faster but less powerful approaches have been used and developed further instead. We present a simple linear-time algorithm for lospre for structured programs that can also handle some more general scenarios compared to previous approaches. We prove that our approach is optimal and that the runtime is linear in the number of nodes in the control-flow graph. The condition on programs of being structured is automatically true for many programming languages and for others, such as C, is equivalent to a bound on the number of \texttt{goto} labels per function. An implementation in a mainstream C compiler demonstrates the practical feasibility of our approach. Our approach is based on graph-structure theory and uses tree-decompositions.
We also show that, for structured programs, the runtime of deterministic implementations of the previously known MC-PRE and MC-SSAPRE algorithms can be bounded by $O(n^{2.5})$, improving the previous bounds of $O(n^3)$.
\end{abstract}



\section{Introduction}\label{lospre:Introduction}

Redundancy elimination is a technique commonly used in current optimizing compilers. Even early optimizing compilers had common subexpression elimination (CSE) for straight-line code. Global common subexpression elimination (GCSE)\index{global common subexpression elimination}~\cite{gcse} extended this to the whole control-flow graph (CFG). Partial redundancy elimination (PRE)\index{partial redundancy elimination}~\cite{pre} generalized GCSE (and some other techniques such as loop-invariant code-motion (LICM)) to computations that are redundant only on some paths in the CFG. Improvements led to lazy code-motion (LCM)\index{lazy code motion}~\cite{lcm}, which is a lifetime-optimal variant of PRE, i.\,e.\ the lifetimes of introduced temporary variables are minimized, which is important to keep  register pressure low. Another improvement to PRE is speculative PRE (SPRE)~\cite{spre0,spre1,spre2}, which can increase the number of computations on some paths in order to reduce the total number of computations done (based on profiling information). The natural improvement is combining the advantages of LCM and SPRE, resulting in lifetime-optimal SPRE (lospre), which was first achieved by the min-cut-PRE (MC-PRE) algorithm~\cite{MC-PRE}.

However, MC-PRE relies on solving a weighted minimum cut problem on directed graphs. This problem seems to be harder than its equivalent on undirected graphs. The fastest known algorithm~\cite{Karger} is randomized, and will likely give a result in $O(n^2\log^3(n))$. Typical implementations use deterministic algorithms resulting in cubic runtime. Some researchers consider this too much for some applications, in particular for just-in-time compilation and thus developed faster approaches that are not optimal, but perform close to lospre for many real-world scenarios~\cite{fppre,ispre}. MC-SSAPRE~\cite{MC-SSAPRE} is a newer optimal approach that works for programs in static single assignment form. In practice it is often faster than MC-PRE, but it has the same worst case complexity, since it also needs to solve a weighted minimum cut problem on directed graphs. MC-SSAPRE is also considered hard to implement~\cite{PRELLVM}.

An alternative to lospre is complete PRE (ComPRE)~\cite{ComPRE}, which completely eliminates redundancies and is lifetime-optimal. It is not speculative, and eliminates more dynamic computations than speculative approaches, which can offer benefits when optimizing for code speed in some scenarios. However, it is not suitable for optimization for code size, and it restructures the CFG, resulting in additional conditional jumps being introduced. For most common architectures, conditional jumps are far more expensive in terms of speed than typical computations~\cite{Riseman1972,Eyerman2006}, so the gain from eliminating redundancy at the cost of introducing conditional jumps is questionable.

We present a simple approach to lospre, which does not sacrifice optimality, and runs in linear time. It is based on graph-structure theory, in particular the bounded tree-width of control-flow graphs of structured programs. It has been implemented in a mainstream C compiler and has low compilation time overhead. It also generalizes lospre further than previous approaches, by allowing trade-offs between costs from computations and costs from variable lifetimes. Unlike previous approaches, our approach can also consider benefits from reduction in the life-times of operands of redundant expressions, while previous approaches only considered costs from the newly introduced temporary variable.

\section{\label{test}Problem Description}\label{problem}

Programs tend to contain redundancies, and eliminating them is an important goal in optimizing compilers. Figure \ref{Cfunc} shows a function written in the C programming language. It will be transformed by compilers into a form of intermediate code. Figure \ref{Ccfg} shows the control-flow graph and intermediate code for this function as it is used in the SDCC \cite{sdcc} compiler before redundancy elimination. The array-index style access got transformed into a sequence of operations, first multiplying the index i by the size of \texttt{long} (the multiplication already was further transformed into a left-shift), followed by an addition of the result to the array base address and then a read from memory. Redundancy elimination would move these operations before the branch, as can be seen in Figure \ref{Ccfglospre}.

\suppressfloats

\begin{figure}
\begin{lstlisting}
#include <stdbool.h>

extern long *a;
extern long c;

void f(bool b, int i)
{
    if(b)
        c = a[i] + 8;
    else
        c = a[i] - 13;
}
\end{lstlisting}
\caption{\label{Cfunc}Function written in C}
\end{figure}

\begin{figure}
\centerline{\includegraphics[scale=0.36]{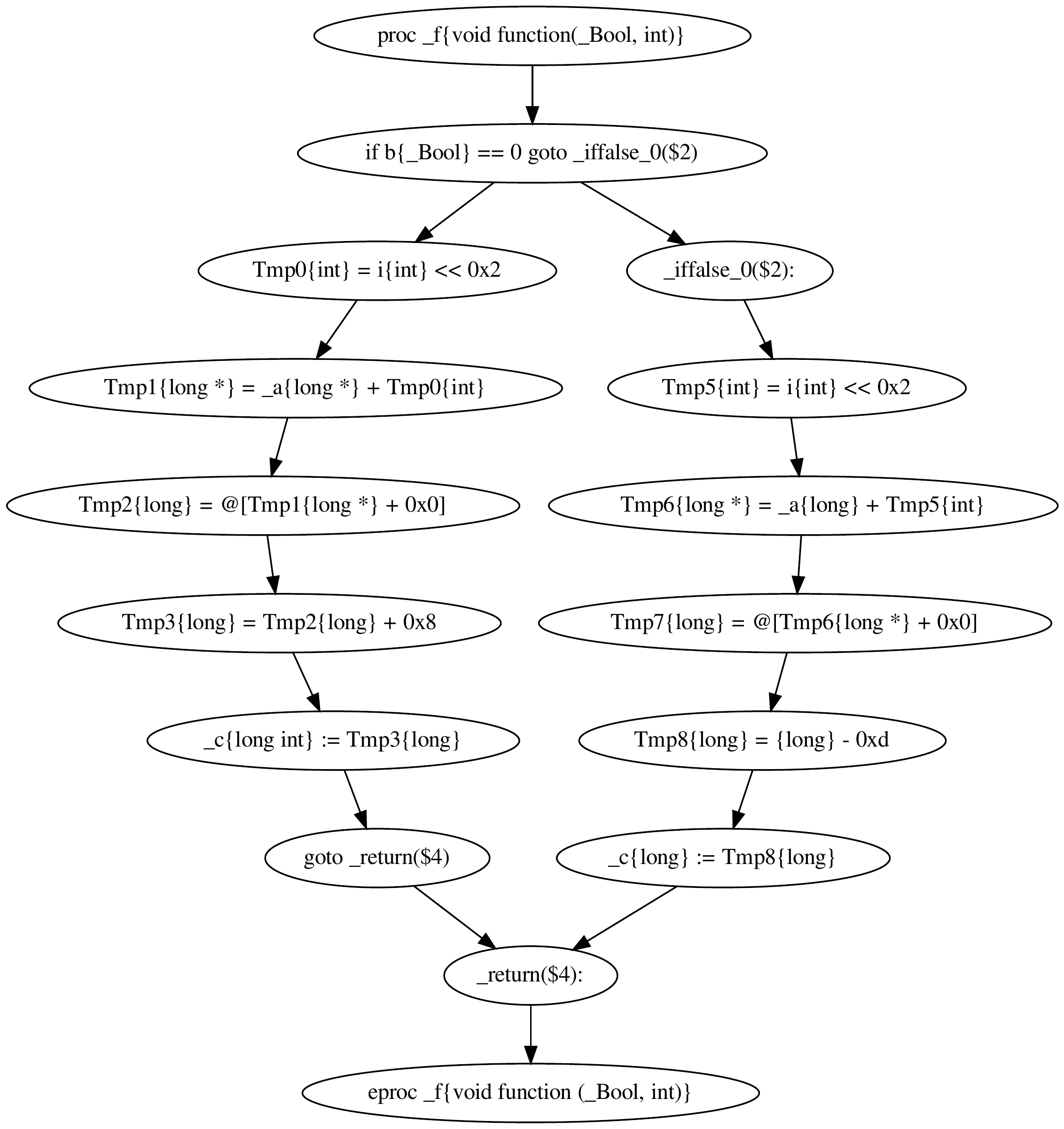}}
\vspace{2mm}
\caption{\label{Ccfg}CFG before redundancy elimination}
\end{figure}

\begin{figure}
\centerline{\includegraphics[scale=0.36]{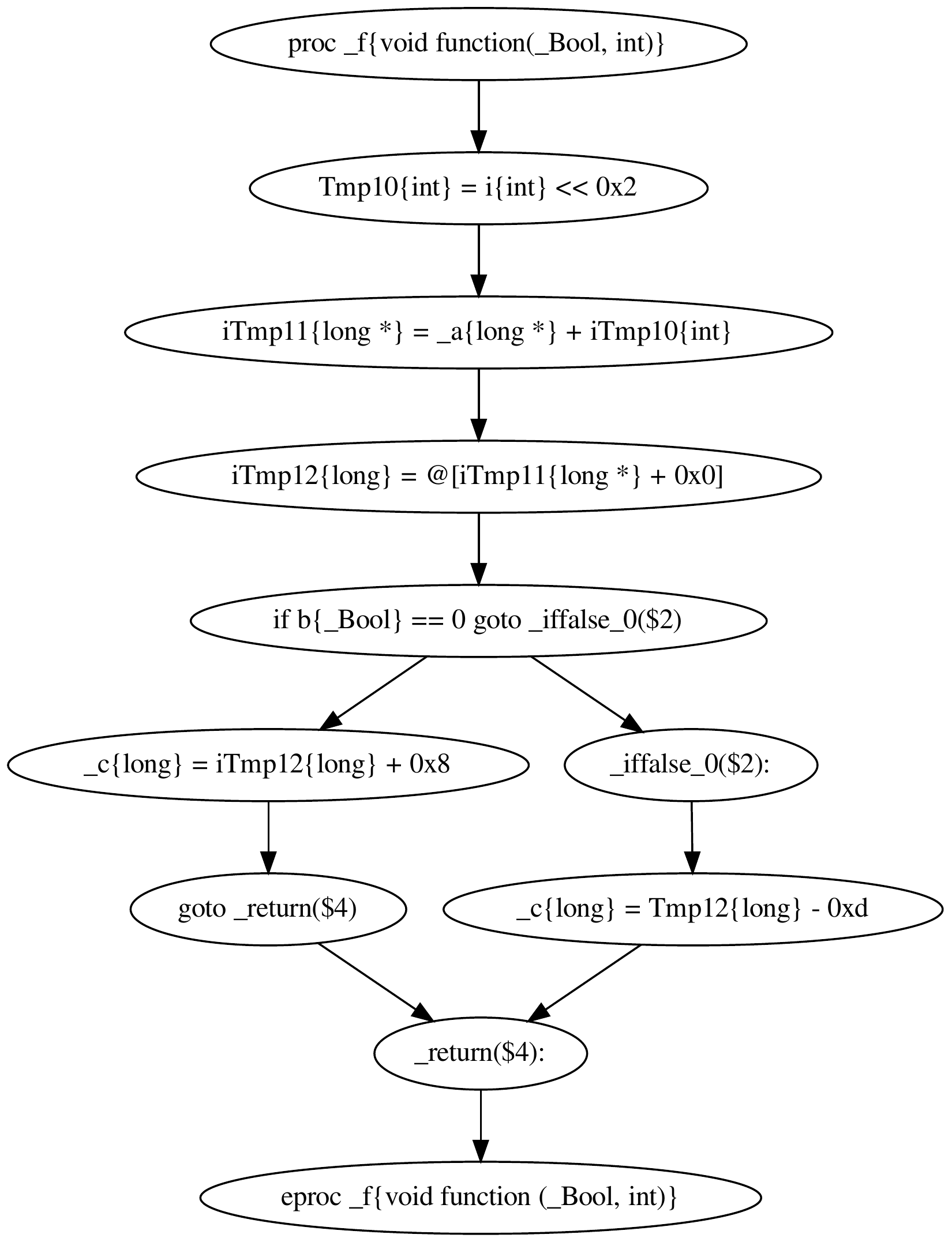}}
\vspace{2mm}
\caption{\label{Ccfglospre}CFG after redundancy elimination}
\end{figure}

We will now formally define the control-flow graph as we use it in our approach. In particular, it is a weighted graph to allow the representation of costs from calculations and variable lifetimes. Often, such costs will be e.\,g.\ represented by natural numbers or something similar, but we want to keep things a bit more general for now:

\begin{definition}[Con\-trol-flow graph (CFG)]
A \emph{con\-trol-flow graph}\index{control-flow graph} is a weighted directed graph $(V, E, c, l)$ with node set $V$ and edge set $E \subseteq V^2$ and weight functions $c\colon E \to \mathcal{K}$ and $l\colon V \to \mathcal{K}$ for some ordered $\mathcal{K}$ that has addition. There is a unique source (node without predecessors).
\end{definition}

For simplicity we will sometimes ignore the weights and treat $G$ as the directed graph $(V, E)$. The nodes of the CFG correspond to instructions in a program. We chose this notion over the more common one of using basic blocks as nodes in the CFG, since it simplifies the discussion of our approach a bit and makes it subsume CSE as well. For applications where compilation speed is essential, it is easy to reformulate our approach to use basic blocks, and do CSE as a pre-processing step. The weight function $c$ gives the cost of subdividing an edge and inserting a computation there. It depends on the optimization goal. E.\,g.\ when optimizing for speed or energy consumption, execution frequencies (estimated or obtained from a profiler) would be used, and $\mathcal{K}$ could be the set of possible values of a floating-point data type. When optimizing for code size one could use a constant $c \in \mathbb{N}$ and $\mathcal{K} = \mathbb{N}$ instead. The weight function $l$ gives the cost of having a new temporary variable alive at a node. When just minimizing the life-time, a constant can be used. More sophisticated approaches to $l$ could take other aspects, such as register pressure at the nodes, into account. This latter aspect, and the possibility of handling costs from calculations and lifetimes in a unified way is something previous approaches, such as MC-PRE and MC-SSAPRE could not do.

For a given expression the set of nodes in the CFG where it is calculated is called the \emph{use set}\index{use set} $\mathcal{U} \subseteq V$. In redundancy elimination techniques, such calculations from $\mathcal{U}$ are replaced by assignments from a new temporary variable, which is initialized by new calculations. For each node in the CFG, we decide whether the new temporary variable should be alive there. We call the set of such nodes the \emph{life set}\index{life set} $\mathcal{L}  \subseteq V$. There also can be invalidating nodes, which we have to be careful about. They invalidate the result of the calculation in the expression: E.\,g.\ when the expression we want to replace is $a + b$, the node where the instruction $a = 7$ is done would be invalidating. We call the set of such nodes the \emph{invalidation set}\index{invalidation set} $\mathcal{I} \subseteq V$. We always consider the source and sinks (nodes without successors) of the CFG to be invalidating, since we have no knowledge about how variables might change before or after our program. Depending on these three sets some edges need to be subdivided and new calculations inserted. We call the set of such edges the \emph{calculation set}\index{calculation set}
\begin{displaymath}
\mathcal{C}(\mathcal{U}, \mathcal{L}, \mathcal{I}) \coloneqq \left\{(x, y) \in E\ \middle|\ x \notin \mathcal{L} \smallsetminus \mathcal{I}, y \in \mathcal{U} \cup \mathcal{L}\right\}.
\end{displaymath}

\begin{definition}[lospre]
Given a CFG $(V, E, c, l)$, use set $\mathcal{U} \subseteq V$ and invalidation set $\mathcal{I} \subseteq V$, the problem of \emph{lospre} is to find a life set $\mathcal{L} \subseteq V$, such that the cost
\begin{displaymath}
\sum_{e \in \mathcal{C}(\mathcal{U}, \mathcal{L}, \mathcal{I})}c(e) + \sum_{v \in \mathcal{L}}l(v)
\end{displaymath}
is minimized.
\end{definition}

For the typical lospre application, one could use $\mathcal{K} = \mathbb{Z}^2$ with lexicographical ordering. Optimizing for code size using $c(e) = (1, 0)$ or optimizing for speed using $c(e) = (p(e), 0)$, where $p$ gives an execution probability estimated using a profiler. $l = (0, 1)$ would then guarantee the lifetime-optimality.

Sometimes \emph{safety}\index{safety} is required. Safety means that no calculations may be done on values on which the original program would not perform them. In general, safety is not desirable, since it restricts choices and thus results in less optimization. However, on some architectures division by zero results in undesirable behaviour; in this case we need safety when the expression is division and we cannot predict the value of the divisor. Other examples would be memory reads from a calculated address on architectures with memory management (where reads from invalid addresses could result in a SIGSEGV terminating the program), accesses to variables declared using C \texttt{volatile}, or I/O accesses. The safety requirement can be handled by using a different invalidation set $\mathcal{I}'$ in place of $\mathcal{I}$ in lospre.

\begin{definition}[safety]
Given a CFG $(V, E)$, use set $\mathcal{U} \subseteq V$ and invalidation set $\mathcal{I} \subseteq V$, the problem of \emph{safety} is finding a set $\mathcal{I}' \supseteq \mathcal{I}$ of minimal size, such that no node outside of $\mathcal{I}'$ lies on a path between two nodes in $\mathcal{I}$ that does not contain a node in $\mathcal{U}$.
\end{definition}

Once we have $\mathcal{I}'$, we can use it in place of $\mathcal{I}$ in lospre.

\begin{figure}
\centerline{\includegraphics[scale=0.357]{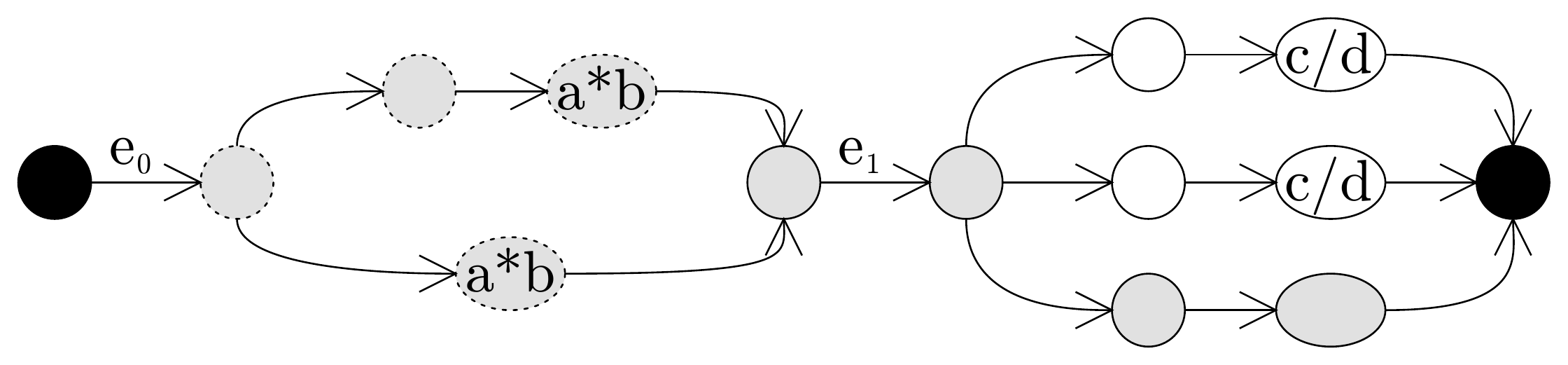}}
\caption{\label{ExampleRedundancy}Simplified CFG with redundancies}
\end{figure}

Figure \ref{ExampleRedundancy} shows a CFG with redundancies. $a * b$ is calculated in two places. Even a simple redundancy elimination technique, such as GCSE, would split $e_0$ and place the calculation there. For lospre we would have $\mathcal{U}$ consisting of the two nodes where $a * b$ is used, no invalidating nodes except for sink and source, thus $\mathcal{I}$ would consist of the black nodes only. The life set found by lospre would consist of the nodes with dashed contours, and thus $\mathcal{C}(\mathcal{U}, \mathcal{L}, \mathcal{I}) = \{e_0\}$.

$c / d$ is a slightly more complicated case: When optimizing for code size, and when safety is not required for division, one might want to split $e_1$ and do the calculation there. But when division requires safety this is not possible ($\mathcal{I}'$ would consist of the black and grey nodes). When optimizing for speed we will not want to pay the cost of calculating $c / d$ when it is not needed.

\section{Structured Programs\label{structured}}

Our approach is based on tree-decompositions. Tree-de\-com\-posi\-tions \cite{Halin, GMIII} have commonly been used to find polynomial algorithms on restricted graph classes for many problems that are hard on general graphs. This includes well known problems such as graph coloring and vertex cover. Algorithmic meta-theorems state that large classes of problems can be solved efficiently on such graph classes \cite{Bodlaender,Courcelle,Seese}, but the algorithms from these meta-theorems are often impractical due to huge constants in run-time and prohibitively hard to implement \cite{FrickGrohe,BodlaenderImpl}. Thus, usually there is still a need to find practical algorithms for individual problems, even where an algorithm with the same asymptotic runtime bound is already provided by a meta-theorem.

\begin{definition}[Tree-decompostion]\label{definition:tree-decomposition}
	A \emph{tree-decomposition} of a graph $G = (V, E)$ is a pair  $(T, \chi)$, consisting of a tree $T$ and a mapping $\chi$ from the node set of the tree $T$ into the power set of the nodes of the graph $G$. Such that the following conditions hold:
	\begin{itemize}
		\item For each $v \in V$ there exists a node $t$ of $T$ with $v \in \chi(t)$.
		\item For each edge $e \in E$ there exists a node $t$ of  $T$ so that both endpoints of $e$ are in $\chi(t)$.
		\item For each $v \in V$ the set $\{t \textrm{ node of } T \mid v \in \chi(t)\}$ is connected in $T$.
	\end{itemize}
	The $\chi(t)$ are called the \emph{bags}\index{bag} of the tree-decomposition.

	The \emph{width}\index{width of a tree-decomposition} of a tree-decomposition $(T, \chi)$ is
	\begin{displaymath}
		\operatorname{w}(T, \chi) \coloneqq \max\left\{\left|\chi(t)\right|-1\ \middle|\ t \textrm{ node of } T\right\}.
	\end{displaymath}

	The \emph{tree-width of $G$}\index{tree-width} is
	\begin{displaymath}
		\tw(G) \coloneqq \min\left\{\operatorname{w}(T, \chi)\ \middle|\ (T, \chi) \text{ is a tree-dec.\ of }G\right\}.
	\end{displaymath}
\end{definition}

Intuitively, tree-width indicates how tree-like a graph is. Nontrivial trees have tree-width one.
Cliques on $n \geq 1$ nodes have tree-width $n - 1$.
Series-parallel graphs have tree-width at most $2$.
Figure \ref{ExDec} gives an example of a tree-decomposition of minimum width for the graph in Figure \ref{ExGraph}. Tree-decompositions are usually defined for undirected graphs, and our definition of a tree-decomposition for a directed graph is equivalent to the tree-decomposition of the graph interpreted as undirected. We will use the notation $|T|$ for the number of nodes in the tree $T$.

\begin{figure}
\centerline{
\subfigure[\label{ExGraph}Graph]{
\includegraphics[scale=0.9]{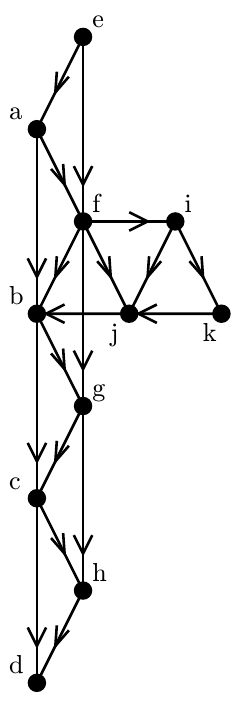}
}
\subfigure[\label{ExDec}Tree-dec.]{
\includegraphics[scale=0.9]{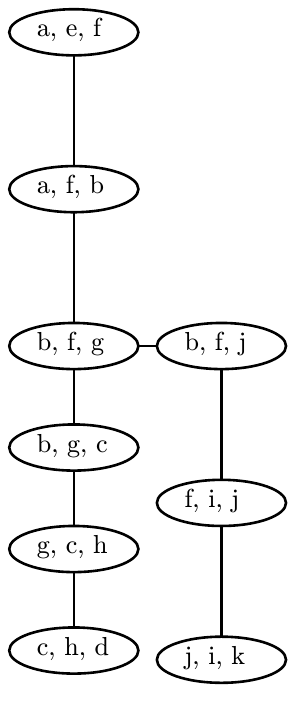}
}
\subfigure[\label{ExNDec}Nice tree-dec.]{
\includegraphics[scale=0.9]{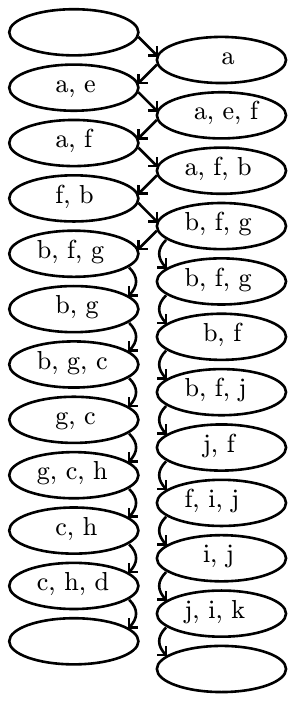}
}
}
\caption{\label{ExampleDecomposition}A graph and tree-decompositions of minimum width}
\end{figure}

\begin{definition}[Structured program]
Let $\mathfrak{k} \in \mathbb{N}$ be fixed. A program is called $\mathfrak{k}$-\emph{structured}, if its CFG has tree-width at most~$\mathfrak{k}$.
\end{definition}

Programs written in Algol or Pascal are $2 + \mathfrak{g}$-structured, if the number of labels targeted by goto statements per function does not exceed $\mathfrak{g}$, Modula-2 programs are 5-structured~\cite{Thorup1998}. Programs written in C are $(7 + \mathfrak{g})$-structured if the number of labels targeted by gotos per function does not exceed $\mathfrak{g}$~\cite{Ctree}. Similarly, Java programs are $(6 + \mathfrak{g})$-structured if the number of labels targeted by labeled breaks and labeled continues per function does not exceed $\mathfrak{g}$~\cite{Java}. Ada programs are $(6 + \mathfrak{g})$-structured if the number of labels targeted by gotos and labeled loops per function does not exceed $\mathfrak{g}$~\cite{Ada}. Coding standards tend to place further restrictions, resulting e.\,g.\ in C programs being $5$-structured when adhering to the widely adopted MISRA-C:2004~\cite{MISRA} standard. Empirically, tree-width greater than $3$ is very rare in programs~\cite{Java,Ctree}.
There are various efficient ways of obtaining tree-decompositions of small width for control-flow graphs~\cite{Thorup1998,Ctree,TreewidthcomputationsIII}.

Often proofs and algorithms are easier to describe, understand and implement when using nice tree-decompositions. Nice tree-decompositions have been used at least since 1994 \cite{nice}, with minor variation in the definitions, sometimes relaxing conditions where they are not needed.  We use the following Definition~\ref{lospre:nicetreedec}, in which the tree in the tree-decomposition is directed. A root in a directed tree is a unique node $r$, such that for every node $n$ in the tree there is a directed path from $r$ to $n$.

\begin{definition}[\label{lospre:nicetreedec}Nice Tree-Decomposition]
A tree-de\-com\-posi\-tion $(T, \chi)$ of a directed graph $G$ is called \emph{nice}, if
\begin{itemize}
\item $T$ has a root $t$, $\chi(t) = \emptyset$.
\item Each node $i$ of $T$ is of one of the following types:
\begin{itemize}
\item Leaf, no children, $\chi(i) = \emptyset$.
\item Introduce node, has one child $j, \chi(j) \subsetneq \chi(i)$.
\item Forget node, has one child $j, \chi(j) \supsetneq \chi(i), |\chi(j) \smallsetminus \chi(i)| = 1$.
\item Join node, has two children $j_1, j_2, \chi(i) = \chi(j_1) = \chi(j_2)$.
\end{itemize}
\end{itemize}
\end{definition}

This terminology is inspired by a bottom-up approach, which is how most algorithms on tree-decompositions (including the ones we propose) work.
Given a tree-decomposition, a nice one of the same width can be found easily in linear time. Figure \ref{ExNDec} shows a nice tree-de\-com\-posi\-tion for the graph in Figure \ref{ExGraph}.

In compiler construction approaches based on the bounded tree-width of structured programs are currently found in register allocation~\cite{Thorup1998,Bodlaender,KrauseRalloc}, data-flow analysis~\cite{Dominators} and bank-selection instruction placement~\cite{Naddr}.

One approach to improving the runtime of MC-PRE would be to replace the min-cut step by one based on tree-de\-com\-posi\-tions. However this would not affect other parts of MC-PRE, and thus not yield a linear time variant of MC-PRE. Also it would not simplify the MC-PRE algorithm, and not allow us to extend lospre to better handle lifetime-optimality. We therefore instead replace all of MC-PRE by an approach based on tree-decompositions.

\section{lospre in linear time}

Our approach uses dynamic programming~\cite{dynprog}, bottom-up along a nice tree-decomposition of minimum width of the CFG; as noted above, there are established methods for obtaining this tree-decomposition. Otherwise, our approach only uses elementary operations.

Let $G = (V, E, c, l)$ be the control-flow graph as in Section~\ref{problem}. Let $(T, \mathcal{X})$ be a nice tree-decomposition of minimum width of $G$ with root $t$.
Let $T_i$ be the set of nodes in bags in the subtree rooted at node $i$ of $T$, excluding the nodes in $\chi(i)$.
Let $S$ be the function that gives the minimum possible costs on edges and nodes covered in the subtree rooted at node $i$ of $T$, excluding edges and nodes in $\chi(i)$. $S$ depends on where in $\chi(i)$ the new temporary variable is alive, which is captured by the functions $f\colon \chi(i) \to \{0, 1\}$. For these $f$ (and any function from a  $X \subseteq V$ to $\{0, 1\}$), we define a local version of the calculation set:
\begin{gather*}
\mathcal{C}(\mathcal{U}, f, \mathcal{I}) \coloneqq \\ \left\{(x, y) \in E \ \middle|\ x \notin f^{-1}(1) \smallsetminus \mathcal{I}, y \in \mathcal{U} \cup f^{-1}(1)\right\}.
\end{gather*}
\begin{gather*}
T_i \coloneqq \\ \left\{v \in (\chi(j) \smallsetminus \chi(i))\ \middle|\ j \text{ in the subtree of } T \text{ rooted at } i\right\},\\
\end{gather*}
\begin{gather*}
S(i, f) \coloneqq \\ \min_{\substack{g\colon V \to \{0, 1\}\\g|_{\chi(i)} = f|_{\chi(i)}}} \left\{ \sum_{v \in (T_i \cap g^{-1}(1))}\!\!\!\! l(v) + \!\!\!\!\sum_{e \in \mathcal{C}(\mathcal{U}, g, \mathcal{I}) \cap T_i^2}\!\!\!\! c(e) \right\}.
\end{gather*}
At the root $t$ of $T$, this function $S$, and the corresponding $g$ are what we want, since this $g$ gives the life-set that is the solution to lospre:
\begin{multline*}
S(t, f) = 
\!\!\!\!\min_{\substack{g\colon V \to \{0, 1\}\\g|_{\chi(t)} = f|_{\chi(t)}}}\!\! \left\{ \sum_{v \in (T_t \cap g^{-1}(1))}\!\!\!\!\!\! l(v) + \!\!\!\! \sum_{e \in \mathcal{C}(\mathcal{U}, g, \mathcal{I}) \cap T_t^2}\!\!\!\! c(e) \right\} = \displaybreak[0]\\
\min_{\substack{g\colon V \to \{0, 1\}\\g|_\emptyset = f|_\emptyset}} \left\{ \sum_{v \in (g^{-1}(1))}l(v) + \sum_{e \in \mathcal{C}(\mathcal{U}, g, \mathcal{I})} c(e) \right\} = \displaybreak[0]\\
\min_{g\colon V \to \{0, 1\}} \left\{ \sum_{v \in (g^{-1}(1))}l(v) + \sum_{e \in \mathcal{C}(\mathcal{U}, g, \mathcal{I})} c(e) \right\} = \displaybreak[0]\\
\min_{\mathcal{L} \subseteq V} \left\{ \sum_{e \in \mathcal{C}(\mathcal{U}, \mathcal{L}, \mathcal{I})}c(e) + \sum_{v \in \mathcal{L}}l(v) \right\}.
\end{multline*}

To calculate $S$, we define a function $s$, and then proceed to show that $s = S$ and that $s$ can be computed in linear time.
We define $s$ inductively depending on the type of $i$:
\begin{itemize}
\item Leaf: $s(i, f) \coloneqq 0.$
\item Introduce node with child $j$: $s(i, f) \coloneqq s(j, f|_{\chi(j)}).$
\item Forget node with child $j, \chi(j) \smallsetminus \chi(i) = \{v\}$: $s(i, f) \coloneqq$
\begin{displaymath}
\min_{\substack{g\colon \chi(j) \to \{0, 1\}\\g|_{\chi(i)} = f}} \!\! \left\{s(j, g) + f(v)l(v) + \!\!\!\!\!\!\!\!\!\!\!\! \sum_{\substack{e \in \mathcal{C}(\mathcal{U}, f, \mathcal{I}) \cap \\ ((\{v\} \times V) \cup (V \times \{v\}))}} \!\!\!\!\!\!\!\!\!\!\!\! c(e)\right\}.
\end{displaymath}
\item Join node with children $j_1$ and $j_2$: $s(i, f) \coloneqq s(j_1, f) + s(j_2, f).$
\end{itemize}

\begin{lemma}\label{correctnesslemma}
$s(i, f)$ gives the minimum possible costs on edges and nodes covered in the subtree rooted at node $i$ of $T$, excluding edges and nodes in $\chi(i)$, i.\,e.\ $s = S$.
\end{lemma}

\begin{proof}
By induction we can assume that the lemma is true for all children of node $i$ of $T$.

Case 1: $i$ is a leaf. There are no edges or nodes in the subgraph of $G$ induced by $T_i = \chi(i) \smallsetminus \chi(i) = \emptyset$, thus the cost is zero:
\begin{displaymath}
s(i, f) = 0 = S(i, f).
\end{displaymath}

Case 2: $i$ is an introduce node with child $j$. $T_i = T_j$, since $\chi(i) \subseteq \chi(j)$, thus the cost remains the same:
\begin{displaymath}
s(i, f) = s(j, f) = S(j, f) = S(i, f).
\end{displaymath}

Case 3: $i$ is a forget node with child $j$. $T_i = T_j \cup (\chi(j) \smallsetminus \chi(i)) = T_j \cup \{v\}$, the union is disjoint. Thus we get the correct result by adding the costs for the edges between $\chi(i)$ and $v, \{v\} = (\chi(j) \smallsetminus \chi(i))$ and the lifetime cost for $v$ itself:
\begin{gather*}
s(i, f) = \min_{\substack{g\colon \chi(j) \to \{0, 1\}\\g|_{\chi(i)} = f}}\left\{s(j, g) + f(v)l(v) + \!\!\!\!\!\!\!\!\!\!\!\! \sum_{\substack{e \in \mathcal{C}(\mathcal{U}, f, \mathcal{I}) \cap \\ ((\{v\} \times E) \cup (E \times \{v\}))}} \!\!\!\!\!\!\!\!\!\!\!\! c(e)\right\} = \displaybreak[0]\\
\min_{\substack{g\colon \chi(j) \to \{0, 1\}\\g|_{\chi(i)} = f}}\left\{S(j, g) + f(v)l(v) + \!\!\!\!\!\!\!\!\!\!\!\! \sum_{\substack{e \in \mathcal{C}(\mathcal{U}, f, \mathcal{I}) \cap \\ ((\{v\} \times E) \cup (E \times \{v\}))}} \!\!\!\!\!\!\!\!\!\!\!\! c(e)\right\} = \displaybreak[0]\\
\min_{\substack{g\colon \chi(j) \to \{0, 1\}\\g|_{\chi(i)} = f}}\left\{\min_{\substack{h\colon V \to \{0, 1\}\\h|_{\chi(j)} = g|_{\chi(j)}}} \left\{ \sum_{u \in (T_j \cap h^{-1}(1))}\!\!\!\! l(u) + \!\!\!\!\sum_{e \in \mathcal{C}(\mathcal{U}, h, \mathcal{I}) \cap T_j^2}\!\!\!\! c(e) \right\} + \right.\\\left. f(v)l(v) + \!\!\!\!\!\!\!\!\!\!\!\! \sum_{\substack{e \in \mathcal{C}(\mathcal{U}, f, \mathcal{I}) \cap \\ ((\{v\} \times E) \cup (E \times \{v\}))}} \!\!\!\!\!\!\!\!\!\!\!\! c(e)\right\} = \displaybreak[0]\\
\min_{\substack{g\colon V \to \{0, 1\}\\g|_{\chi(i)} = f}}\left\{\sum_{u \in (T_j \cap g^{-1}(1))}\!\!\!\! l(u) + \!\!\!\!\sum_{e \in \mathcal{C}(\mathcal{U}, g, \mathcal{I}) \cap T_j^2}\!\!\!\! c(e) + \right.\\\left. f(v)l(v) + \!\!\!\!\!\!\!\!\!\!\!\! \sum_{\substack{e \in \mathcal{C}(\mathcal{U}, f, \mathcal{I}) \cap \\ ((\{v\} \times E) \cup (E \times \{v\}))}} \!\!\!\!\!\!\!\!\!\!\!\! c(e)\right\} = \displaybreak[0]\\
\min_{\substack{g\colon V \to \{0, 1\}\\g|_{\chi(i)} = f|_{\chi(i)}}} \left\{ \sum_{v \in (T_i \cap g^{-1}(1))}\!\!\!\! l(v) + \!\!\!\!\sum_{e \in \mathcal{C}(\mathcal{U}, g, \mathcal{I}) \cap T_i^2}\!\!\!\! c(e) \right\} = S(i, f).
\end{gather*}

Case 4: $i$ is a join node with children $j_1$ and $j_2$. $T_i = T_{j_1} \cup T_{j_2}$, since $\chi(i) = \chi(j_1) = \chi(j_2)$. The union is disjoint and there are no edges between $T_{j_1}$ and $T_{j_2}$ in G. Thus we get the correct result by adding the costs from both subtrees:
\begin{displaymath}
s(i, f) = s(j_1, f) + s(j_2, f) = S(j_1, f) + S(j_2, f) = S(i, f).
\end{displaymath}
\end{proof}

\begin{lemma}\label{timelemma}
Given a nice tree-decomposition $(T, \mathcal{X})$ of minimum width of $G$, $s$ can be calculated in time $O(\tw(G)2^{\tw(G)}|T|)$.
\end{lemma}

\begin{proof}
At each node $i$ of $T$ time $O(\tw(G)2^{\tw(G)})$ is sufficient:

Case 1: $i$ is a leaf. There is only one function $f\colon \emptyset \to \{0, 1\}$.

Case 2: $i$ is an introduce node. There are at most $2^{|\chi(i)|} \leq 2^{\tw(G) + 1}$ different $f$ and for each one we use constant time.

Case 3: $i$ is a forget node with child $j$. There are at most $2^{|\chi(i)|}$ different $f$ and for each one we need to consider at most $2^{|\chi(j) \smallsetminus \chi(i)|}$ different $g$ and for each $g$ we need to consider at most $2\tw(G)$ different edges. Thus the total time is in $O(2^{|\chi(i)|} 2^{|\chi(j) \smallsetminus \chi(i)|} 2\tw(G)) = O(\tw(G)2^{|\chi(j)|}) \subseteq O(\tw(G)2^{\tw(G) + 1}) = O(\tw(G)2^{\tw(G)})$.

Case 4: $i$ is a join node. The reasoning from case 2 holds.
\end{proof}

\begin{theorem}\label{lospretheorem}
lospre can be done in linear time for structured programs.
\end{theorem}

\begin{proof}
Given an input program of bounded tree-width we can calculate a tree-decomposition of minimum width in linear time \cite{bodlaenderalg}. We can then transform this tree-decomposition into a nice one of the same width. The linear time for these steps implies that $|T|$ is linear in $|V|$. We then calculate the $s$ as in Lemma \ref{timelemma} above in linear time. Using standard book-keeping techniques we can keep track of which $g$ corresponds to each $f$. The one remaining $s(t, f)$ at the root $t$ of $T$ then gives us the minimum total cost according to Lemma \ref{correctnesslemma}. The corresponding $g^{-1}(1) = \mathcal{L}$ is the solution.
\end{proof}

The safety problem can be solved by a similar approach. This time $f$ denotes which nodes of $G$ are to be added to $\mathcal{I}$ to get $\mathcal{I}'$. We use cost values in $\mathcal{K} = \mathbb{Z} \cup \{\infty\}$. The function $s$ is defined the same as above except for forget nodes. For a forget node $i$, with child $j, \{v\} = \chi(j) \smallsetminus \chi(i)$, it is defined as follows:

\begin{multline*}
s(i, f) \coloneqq\\
\min\left\{s(j, g) \ \middle|\ g|_{\chi(i)} = f\right\} +\\
\begin{cases}
0\text{ if } f(v) = 0\\
\infty\text{ if } f(v) = 1, v \in \mathcal{U}\\
\infty\text{ if } f(v) = 1,\text{ no successor of }v \text{ is in }f^{-1}(1) \cup (\mathcal{I} \setminus \mathcal{U})\\
\infty\text{ if } f(v) = 1,\text{ no predecessor of }v \text{ is in }f^{-1}(1) \cup \mathcal{I}\\
- 1\text{ otherwise}.
\end{cases}
\end{multline*}

With a proof very similar to the one for the previous theorem, we get:

\begin{theorem}
The safety\index{safety} problem can be solved in linear time for structured programs.
\end{theorem}

Together with the previous theorem, this allows us to do lospre in linear time, even when safety is required.

\section{Implementation}

We implemented our approach to lospre in SDCC \cite{sdcc}, since SDCC has infrastructure for handling tree-decompositions due to its tree-decomposition based register allocator~\cite{KrauseRalloc} and bank selection~\cite{Naddr}.
SDCC is a C compiler for embedded systems, targeting the MCS-51, DS390, DS400, HC08, S08, Z80, Z180, Rabbit 2000/3000, Rabbit 3000A, LR35902, TLCS-90, STM8, PIC14 and PIC16  architectures. Our implementation ships in SDCC since version 3.3.0 released in May 2013. The code can be found in the SDCC project's public source code repository.

The implementation minimizes the total number of computations in the three-address code (corresponding to optimizing for code size), using $\mathcal{K} = \mathbb{Z}^2$ with lexicographical ordering, $w(e) = (1, 0)$ and $l = (0, 1)$. It does not yet take information from pointer analysis into account, and thus requires safety for all pointer reads (as we cannot rule out that a pointer points to some memory-mapped I/O on a path where it is not normally read). The tree-decomposition is computed once per CFG and then reused for all expressions. Initially, Thorup's heuristic \cite{Thorup1998} was used to obtain the tree-decomposition. Since SDCC 3.7.0, Thorup's heuristic has been replaced by a different approach~\cite{Ctree}.
As a benchmark, we compiled the Contiki\index{Contiki} operating system \cite{contiki}, version 2.5, consisting of 1083 C functions.

To measure the impact of lospre on compiled programs, we first counted the number of eliminated computations when using no other global redundancy elimination technique. To measure the advantage over the current GCSE implementation in SDCC, we also counted the number of computations eliminated by our lospre implementation when GCSE was run on the programs first. When using lospre as an additional compiler stage after GCSE, lospre was able to reduce the number of computations by 543. With GCSE disabled, lospre was able to reduce the number of computations by 1311; when the safety requirement on reads from calculated addresses was dropped, these numbers increased to 561 and 1329. This shows that even in its current state, our lospre implementation provides a significant advantage over the GCSE implementation used in SDCC. For further data on how lospre can improve a program, we refer the reader to the extensive experimental evaluation using the MC-PRE~\cite{MC-PRE} (e.\,g.\ eliminating $90.13\%$ more non-full redundancies in SPECint2000 compared to LCM, speedup of over $7\%$ compared to LCM in the sixtrack SPECint2000 benchmark) and MC-SSAPRE algorithms~\cite{MC-SSAPRE}; the flexible handling of lifetime costs using the function $l$ in our approach offers the potential for improvements over what MC-PRE and MC-SSAPRE can do.

We used the Callgrind tool of Valgrind \cite{Valgrind} to measure the part of compilation time spent in our lospre implementation. Again the numbers are from compiling the Contiki operating system. Only 1.75\% of the total compilation time was spent in lospre. This is a very low overhead, especially considering that our implementation is not optimized for compilation speed. Speed could easily be improved further, e.\,g.\ by parallelization or by working on a CFG that uses basic blocks as nodes. Of the time spent in our implementation, about 66\% were spent in our lospre algorithm, 29\% in our safety algorithm, and the rest on other tasks, such as obtaining a tree-decomposition and transforming it into a nice tree-decomposition.

The compilation speed speed could be improved further by using working on the block graph and by parallelizing it. It could be improved by using information from pointer analysis. Also, the current implementation in SDCC only provides a complexity advantage over the previously known MC-PRE and MC-SSAPRE algorithms; it will be interesting to see the impact from using our approach in its full generality, i.\,e.\ the possibility to have more complex cost functions, which take register pressure into account and allow a trade-off between computation costs and lifetime costs; the next section will allow a first glimpse on this.

\section{Extending lifetime optimality}

Traditionally, the property of being lifetime optimal only referred to minimization of the life-times of newly introduced temporary variables. However, considering further aspects, such as register pressure and potential reductions in life-times of other variables, can result in better code. In this section we show how to extend our approach to handle these aspects, at the cost of increasing the runtime by a small constant factor; this also serves as an example for other similar extensions.
For calculations that have local variables as operands, we introduce $\mathcal{L}_l$ as the life-time of the left operand and $\mathcal{L}_r$ as the life-time of the right operand, and adjust the cost function accordingly:

\begin{displaymath}
\sum_{e \in \mathcal{C}(\mathcal{U}, \mathcal{L}, \mathcal{I})}c(e) + \sum_{v \in V}l\left(v, \mathcal{L} \cap \{v\}, \mathcal{L}_l \cap \{v\}, \mathcal{L}_r \cap \{v\}\right)
\end{displaymath}

To account for the two additional sets we have to handle we also adjust the functions $f$ and $g$ to give values in $\{0, 1\}^3$ instead of $\{0, 1\}$.

\begin{theorem}
lospre can be done in linear time for structured programs, even under the extended meaning of lifetime-optimality.
\end{theorem}

\begin{proof}
The proof is similar to the one for Theorem \ref{lospretheorem} above. The most notable difference is that compared to the proof of Lemma \ref{timelemma} we have to consider $(2^3)^{|\chi(i)|}$ different $f$ instead of just $2^{|\chi(i)|}$, resulting in total time $O((\tw(G)2^{\tw(G)} + 8^{\tw(G)})|T|)$ instead of $O(\tw(G)2^{\tw(G)}|T|)$ for the calculation of $s$. Since $\tw(G)$ is bounded, this is still linear time.
\end{proof}

We have started extending our implementation accordingly, using register pressure\index{register pressure} (i.\,e.\ the sum over the sizes of all variables alive at an instruction) for $l$. For some test cases we see improvements in the size and runtime of the generated code of up to $2.5$\%. Results also hint at potential for further improvement when taking the availability of registers in the target architecture into account, so this is what we want to investigate next.

\section{MC-PRE and MC-SSAPRE}

To be able to experimentally compare runtimes, we also implemented MC-PRE and ran both our algorithm and MC-PRE on the block graphs obtained from a large number of C functions. As expected, our algorithm and MC-PRE gave the same final results. However, we did not find the big difference in runtimes between our approach and MC-PRE that we expected: MC-PRE runtimes on the benchmarks we used were actually quite good, and we decided to look into the reasons.

In MC-PRE the most important step for the runtime is the weighted minimum cut, since all other steps can be done in quadratic time. A lot of research has gone into implementing weighted minimum cut efficiently. We used the implementation from the boost C++ libraries, which uses uses a deterministic push-relabel max flow algorithm, with highest-label node order, the variant that was found to be the fastest in comparisons~\cite{Cherkassky1994}, and is also used by the MC-PRE authors~\cite{MC-PRE}. For a graph of $n$ nodes and $m$ edges, it has complexity $O(n^2\sqrt{m})$, usually stated as $O(n^3)$. MC-PRE computes the weighted minimum cut on a graph $G_{st} = (N_{st}, E_{st}, W_{st})$. It is easy to show that for a control-flow graph $G = (N, E)$, we have $|N_{st}| \leq 2|N| + 2$ and $|E_{st}| \leq 2|E| + 4 |N|$. Often, $N_{st}$, and $E_{st}$ will be much smaller.
By Definition \ref{definition:tree-decomposition}, for each edge in the graph, there is a bag in the tree-decomposition that contains both ends. For fixed width, size of the tree-decomposition is linear in the size of the graph. We get that for control-flow graphs of bounded tree-width, the number of edges is at most the number of nodes times a constant factor, and thus:

\begin{theorem}
For structured programs, the runtime of MC-PRE implemented using a deterministic push-relabel max flow algorithm is in $O(n^{2.5})$.
\end{theorem}

In MC-SSAPRE, all steps except for the computation of a weighted minimum cut have linear complexity~\cite{MC-SSAPRE}. Using similar arguments as above for MC-PRE, we get:

\begin{theorem}
For structured programs, the runtime of MC-SSAPRE implemented using a deterministic push-relabel max flow algorithm is in $O(n^{2.5})$.
\end{theorem}

\section{Conclusion}

We have presented an algorithm for lospre, which is based on graph-structure theory. Like the earlier MC-PRE and MC-SSAPRE algorithms it is optimal, but it is more general, and much simpler and has much lower time complexity. We have proven that it is optimal and has linear runtime. An implementation in a mainstream C compiler demonstrates the practical feasibility of our approach and the low compilation time overhead.

We also presented improved time complexity bounds for deterministic implementations of MC-PRE and MC-SSAPRE.

\bibliographystyle{lospreplain}
\bibliography{bib}

\end{document}